\documentclass[envcountsect, 11pt]{llncs}

\usepackage[left=3cm,right=3cm,top=3.5cm,bottom=3.5cm]{geometry}
\usepackage{subfigure}
\usepackage{amsmath}
\usepackage{amsfonts}
\usepackage{tikz}
\usepackage{mathdots}
\usetikzlibrary{shapes.misc}
\usepackage[ruled]{algorithm2e}
\usepackage[citecolor=blue, linkcolor=blue, colorlinks=true]{hyperref}
\usepackage{pgfplots}
\usepackage{placeins} 
\usepgfplotslibrary{groupplots}

\pgfmathdeclarefunction{gauss}{2}{%
  \pgfmathparse{1/(#2*sqrt(2*pi))*exp(-((x-#1)^2)/(2*#2^2))}%
}

\newcommand\refTheorem[1]{\hyperref[#1]{Theorem~\ref*{#1}}}
\newcommand\refLemma[1]{\hyperref[#1]{Lemma~\ref*{#1}}}
\newcommand\refCorollary[1]{\hyperref[#1]{Corollary~\ref*{#1}}}
\newcommand\refFigure[1]{\hyperref[#1]{Figure~\ref*{#1}}}
\newcommand\refAlgorithm[1]{\hyperref[#1]{Algorithm~\ref*{#1}}}
\newcommand\refEquation[1]{\hyperref[#1]{(\ref*{#1})}}

\newtheorem{observation}{Observation}

\newcommand{\np}{\mbox{{\normalfont NP}}}

\newcommand{\ppad}{\mbox{{\normalfont PPAD}}}
\newcommand{\optp}{\mbox{{\normalfont OptP}}}
\newcommand{\pspace}{\mbox{{\normalfont PSPACE}}}
\newcommand{\eotl}{{\mbox{{\sc End of line}}}}
\newcommand{\sperner}{{\mbox{{\sc Sperner}}}}
\newcommand{\twodb}{{\mbox{{\sc 2D-discrete-Brouwer}}}}
\newcommand{\threedb}{{\mbox{{\sc 3D-discrete Brouwer}}}}
\newcommand{\qbf}{{\mbox{{\sc QBF}}}}

\newcommand{\hide}[1]{}

\newcommand{\alg}{\ensuremath{\mathcal A}}

\begin{document}
\title{The Complexity of the Path-following Solutions of Two-dimensional Sperner/Brouwer Functions}




\author{Paul W. Goldberg\thanks{Supported by EPSRC under grant EP/K01000X/1}}



\institute{Department of Computer Science\\
University of Oxford\\
 \email{paul.goldberg@cs.ox.ac.uk}}


\maketitle

\begin{abstract}
There are a number of results saying that for certain ``path-following'' algorithms that
solve \ppad-complete problems, the solution obtained by the algorithm is \pspace-complete
to compute. We conjecture that these results are special cases of a much more
general principle, that all such algorithms compute \pspace-complete solutions.
Such a general result might shed new light on the complexity class \ppad.

In this paper we present a new \pspace-completeness result for an interesting
challenge instance for this conjecture.
Chen and Deng~\cite{CD} showed that it is \ppad-complete to find a trichromatic
triangle in a concisely-represented Sperner triangulation.
The proof of Sperner's lemma --- that such a solution always exists --- identifies one solution
in particular, that is found via  a natural ``path-following'' approach.
Here we show that it is \pspace-complete to compute this specific solution,
together with a similar result for the computation of the path-following solution of
two-dimensional discrete Brouwer functions.
\end{abstract}

\section{Preliminaries}

We begin with the definitions needed for the result presented here,
then we mention some related work and motivation.
We consider search problems where any instance $I$ has a set $S_I$ of
associated solutions represented by bit strings of length bounded by some
polynomial in the length of $I$; the challenge is to find an element of $S_I$.
A search problem is in \np\ if there is a polynomial-time algorithm 
that can check, given $I$ and a bit string $s$, whether $s$ belongs to $S_I$.
A search problem is {\em total} if for all $I$, $S_I$ is non-empty.

\begin{definition}\label{def:reduction}
Let $X$ and $Y$ be two \np\ search problems.
A {\em reduction $(f,g)$ from $X$ to $Y$} consists of polynomial-time computable
functions $f$ and $g$, where $f$ maps  instances of $X$ to instances of $Y$,
and $g$ maps solutions of $Y$ to solutions of $X$, such that $g(f(I))$ is a solution of $I$.
\end{definition}

Papadimitriou~\cite{Pap94} introduced the complexity class \ppad,
defined in terms of the problem \eotl\ (Definition~\ref{def:eotl}).
A search problem $X$ belongs to \ppad\ if there exists a polynomial-time
reduction from $X$ to \eotl, and $X$ is \ppad-complete if $X$ belongs to \ppad,
and \eotl\ is reducible to $X$.
\eotl\ is an \np\ total search problem that appears to be hard.

\begin{definition}\label{def:eotl}
An instance of \eotl\ consists of two
directed boolean circuits $P$, $S$, each having $n$ input nodes and $n$
output nodes, and each of size polynomial in $n$. Consider an associated
directed graph $G(P,S)$ on the $2^n$ bit strings, obtained by including arc
$(u,v)$ whenever $S(u)=v$ and $P(v)=u$. The exception to this rule is
that the all-zeroes bit string {\bf 0} is deemed to have no incoming edge.

The problem is to search for a bit string, other than {\bf 0}, that belongs
to exactly one arc (either as a source or a sink).
\end{definition}

Since, by construction, vertices of $G(P,S)$ have at most one incoming arc and at
most one outgoing arc, there is an ``obvious'' exponential-time algorithm
that finds a solution to \eotl\ as follows.
Starting at the vertex {\bf 0}, compute $S({\bf 0})$, $S(S({\bf 0}))$,
$S(S(S({\bf 0})))$, $\cdots$, until a vertex is found that has no outgoing arc.
We call this the path-following approach.
The other natural exponential-time algorithm is lexicographic search.
This paper aims to develop a deeper understanding of the properties
of path-following solutions, and how they contrast with lexicographic
search solutions. Notice that:
\begin{itemize}
\item \eotl\ is a {\em total} search problem: any instance $(P,S)$ has a solution,
and the proof of existence identified the path-following solution.
\item Moreover, a solution to an instance $(P,S)$ is allowed to be {\em any} degree-1
vertex (other than {\bf 0}); as a consequence, \eotl\ is a search problem that belongs to \np\
(it is easy to use $P$ and $S$ to check that a given solution is valid).
\item Although the unique path-following solution can be easily verified as a valid solution,
its status as the path-following solution cannot be efficiently verified.
\end{itemize}

We continue with the definitions of the \ppad-complete problems referred to in the title,
using essentially the same definitions given in~\cite{CD} (although we prefer to describe
them in terms of boolean circuits rather than Turing machines);
we will see that these problems also have natural path-following algorithms.

\begin{definition}
An instance of \twodb\ consists of a directed Boolean circuit $C$ having $2n$
input nodes and 2 output nodes; $C$ is of size polynomial in $n$. We assume that
the values taken by the output nodes represent three colour values $\{0,1,2\}$
(for example, by taking $00$ and $01$ to represent $0$, $10$ to represent 1,
and $11$ to represent $2$). $C$ represents a colouring of the integer grid points
$(x,y)$, $0\leq x,y < 2^n$ as follows.
The input nodes represent the coordinates of point $(x,y)$,
and the outputs represent its colour.
We impose the following {\em boundary condition}, that all
points $(0,y)$ must have colour 1, points $(x,0)$ other than the origin have colour 2,
and other points of the form $(x,2^n-1)$ or $(2^n-1,y)$ have colour 0.

The problem is to search for a unit square whose vertices contain points with all 3 colours.
\end{definition}

\twodb\ has an algorithm (given in detail in Definition~\ref{def:brouweralgo})  that can be described informally as follows.
Notice that a discrete Brouwer function divides the grid into a number of coloured regions,
where regions are connected components of points having the same colour, that can be reached
from each other via unit-length edges, or diagonal edges connecting $(x,y)$ with $(x+1,y-1)$.
We seek a point where all three colours are close together.
The boundary condition identifies a known point where regions having colours 1 and 2
are adjacent. Suppose that we trace the boundary of these two regions; then we either encounter
the third colour, or we exit the rectangle. However, the boundary condition says that there is
only one boundary point where colours 1 and 2 are adjacent, so in fact we encounter the
third colour.

Figure~\ref{fig:redex} (left-hand side) shows an example of a simple discrete Brouwer
function, together with the path followed by the above algorithm.

\begin{definition}
An instance of \sperner\ consists of a colouring of the set of non-negative
integer grid points $(x,y)$ satisfying $x+y<2^n$.
The colouring is specified using a circuit $C$ of size polynomial in $n$,
taking as input the coordinates $(x,y)$ and outputting a colour in $\{0,1,2\}$,
and satisfying the boundary conditions that no point with $x=0$ gets
colour $0$, no point with $y=0$ gets colour $1$, and no point with
$x+y=2^n-1$ gets colour 2. Consider the triangulation of the grid obtained
by connecting pairs of points in this domain that are distance 1 apart,
also all pairs of points $\{(x,y),(x-1,y+1)\}$.

The problem is to search for a ``trichromatic triangle'', that is, one of these triangles having 3
distinct colours at its vertices.
\end{definition}

Next, we review a straightforward reduction from \twodb\ to \sperner;
the reduction serves to identify some of the notation we use.
The reduction is essentially a slight simplification of one presented in~\cite{CD};
we will use it later to show that the path-following solution of
\sperner\ (obtained by the algorithm of Definition~\ref{def:sperneralgo})
is also \pspace-complete to compute.

\paragraph{Reducing \twodb\ to \sperner.}{
Let $I$ be an instance of \twodb\ with size parameter value $n$.
$f$ constructs an instance $I'$ of \sperner\ with size parameter value $n+2$,
where every point $(x,y)$ in $I$ has 4 corresponding points in $I'$ coloured as follows
(Figure~\ref{fig:redex} shows an example):
\begin{itemize}
\item{In $I'$, $(2x,2y)$, $(2x+1,2y)$, and $(2x,2y+1)$ receive the same colour as point $(x,y)$ in $I$.
Point $(2x-1,2y-1)$ receives the colour of $(x,y)$ in $I$, for all $x,y>0$.}
\item Points on the boundary obey the boundary condition of \sperner\ as follows:
the origin is coloured 2; all points $(0,y)$ with $y>0$ are coloured 1;
all points $(x,0)$ with $x<2^{n+1}$ are coloured 2; with $x\geq 2^{n+1}$ are coloured 0.
\item Remaining points are coloured $0$.
\end{itemize}

$g$ takes the triangle $t$ that constitutes a solution of $I'$,
and outputs the square in $I$ that, if all its coordinates were doubled,  contains $t$.
}

\begin{figure}
{\centering
\begin{tikzpicture}[scale=0.5]


\draw[step = 2,black] (0,0) grid (6,6);
\foreach \y in {0,2,4,6}{
  \foreach \x in {0,2,4,6}{
    {\node at (\x,\y){\textcolor{white}{\LARGE $\bullet$}};{\node at (\x,\y){\textcolor{black}{\LARGE $\circ$}};
        }}}}

\foreach \y in {0,2,4,6}{\node at (0,\y){\textcolor{gray}{\LARGE $\bullet$}};}
\foreach \x in {2,4,6}{\node at (\x,0){\textcolor{black}{\LARGE $\bullet$}};}
\node at (2,2){\textcolor{black}{\LARGE $\bullet$}};
\node at (4,4){\textcolor{black}{\LARGE $\bullet$}};
\node at (2,4){\textcolor{gray}{\LARGE $\bullet$}};
\node at (4,2){\textcolor{gray}{\LARGE $\bullet$}};

\foreach \y [count=\yi] in {0,...,15}{
    \pgfmathtruncatemacro\xend{15-\y}
    \draw (11,\y)--(11+\xend,\y); \draw (11+\y,0)--(11+\y,\xend); \draw(11+\y,0)--(11,\y);}

\foreach \y [count=\yi] in {0,...,15}{
    \pgfmathtruncatemacro\xend{15-\y}
  \foreach \x [count=\xi] in {0,...,\xend}{
    {\node at (11+\x,\y){\textcolor{white}{\LARGE $\bullet$}};   }}}

\foreach \y [count=\yi] in {0,...,15}{
    \pgfmathtruncatemacro\xend{15-\y}
  \foreach \x [count=\xi] in {0,...,\xend}{
    {\node at (11+\x,\y){\textcolor{black}{\LARGE $\circ$}};   }}} 

\foreach \y in {1,...,15}{\node at (11,\y){\textcolor{gray}{\LARGE $\bullet$}};}
\foreach \x in {11,...,18}{\node at (\x,0){\textcolor{black}{\LARGE $\bullet$}};}
\foreach \y in {2,4,6}{\node at (12,\y){\textcolor{gray}{\LARGE $\bullet$}};}
\foreach \x in {12,...,16}{\foreach \y in {1,...,5}{\node at (\x,\y){\textcolor{gray}{\LARGE $\bullet$}};}}
\foreach \x/\y in {12/1,13/1,13/2,13/3,14/2,14/3,15/1,15/4,15/5,16/4}{\node at (\x,\y){\textcolor{black}{\LARGE $\bullet$}};}

\foreach \x/\y in {12/5,14/5,16/5,16/3,16/1}{\node at (\x,\y){\textcolor{white}{\LARGE $\bullet$}};}
\foreach \x/\y in {12/5,14/5,16/5,16/3,16/1}{\node at (\x,\y){\textcolor{black}{\LARGE $\circ$}};}

\draw[->,black,ultra thick,rounded corners,dashed] (1,-1)--(1,3)--(3,3)--(3,1)--(5,1);

\draw[->,black,ultra thick,rounded corners,dotted] (10,0.5)--(11.5,0.5)--(11.5,1.5)--(12.5,1.5)--(12.5,3.5)--(14.5,3.5)--(14.5,2.5)--(14.5,1.5)--(13.5,1.5)--(13.5,0.5)--(14.5,0.5)--(14.5,1.5)--(15.3,1.5);

\node at (-0.2,-0.6){$(0,0)$};\node at (10.8,-0.6){$(0,0)$};
\node at (7,6.5){($2^n$-1,$2^n$-1)};\node at (9.8,15.5){(0,$2^{n+2}$-1)};\node at (27,-0.6){($2^{n+2}$-1,0)};

\end{tikzpicture}
\par}
\caption{Example of the reduction from \twodb\ to \sperner.
On the left is a discrete Brouwer function and on the right is the derived Sperner triangulation.
A white point denotes 0, grey denotes 1, black denotes 2.
Dashed arrows show the paths followed by the ``natural'' path-following algorithms.}\label{fig:redex}
\end{figure}


\begin{definition}\label{def:pfa}
A {\em path-following algorithm} \alg\ for a \ppad-complete problem
$X$ is defined in terms of a reduction $(f,g)$ from $X$ to \eotl.
Given an instance $I$ of $X$, \alg\ computes $f(I)$, then
it follows the path in $f(I)$ beginning at the all-zeroes string {\bf 0}, obtaining
a sink $x$ in the \eotl\ graph defined by $f(I)$. The algorithm outputs $g(x)$.
\end{definition}

Path-following algorithms take time exponential in $n$, but polynomial space.
Lexicographic search algorithms are similar in this respect, but have a fundamentally different nature.
The solutions to all known path-following algorithms are \pspace-complete to compute,
while the solutions to lexicographic search problems have lower computational complexity,
essentially in the complexity class \optp\ (Krentel~\cite{K88}) characterising the
complexity of finding the lexicographically first satisfying assignment of an \np\ search problem.
Notice that an incorrect solution $s$ has a short certificate, consisting of a solution $s'$
that precedes $s$ lexicographically.

\begin{paragraph}{}
We next define path-following algorithms for the two problems under consideration.

\begin{definition}\label{def:brouweralgo}
The {\em natural path-following algorithm} for \twodb\ is defined as follows.
\begin{enumerate}
\item Let $S$ be the bottom left-hand square of the grid; by construction $S$
contains vertices coloured 1 and 2.
\item\label{steptwo} If $S$ contains vertices coloured 0, 1, and 2, halt and output $S$
\item Else, let $S'$ be a square adjacent to $S$ that shares an edge with vertices $v_1$ and $v_2$ coloured 1 and 2,
such that $v_1$ is to the left of $v_2$ when viewed from $S$. If 2 such squares exist, select the one
that results in a right turn from the previous direction moved.
\item Set $S=S'$ and return to Step~\ref{steptwo}.
\end{enumerate}
\end{definition}
To see that this is a path-following algorithm according to Definition~\ref{def:pfa},
note that the constructed \eotl\ graph has a node for for every
square in the grid, with an additional node for any square having vertices
coloured 1,2,1,2 in clockwise order, and arcs that connect squares in the
way indicated in the algorithm.
Figure~\ref{fig:redex} (LHS) shows the path leading to the
solution at the bottom right square.

\begin{definition}\label{def:sperneralgo}
The {\em natural path-following algorithm} for \sperner\ is defined as follows.
\begin{enumerate}
\item Given an instance $I$ of \sperner, construct an extended triangulation $T$
by adding a point $p$ to the left, coloured 1, and edges between $p$ and all
grid points $(0,y)$.
\item Let $G$ be a graph where each node corresponds with a triangle of $T$.
\item If triangle $t'$ can be reached from triangle $t$ by crossing an edge
with a vertex coloured 1 on the left and 2 on the right, there is an arc in $G$
between the corresponding nodes of $G$.
\item By construction, the triangle with vertices $\{p,(0,0),(0,1)\}$
has no incoming edge and one outgoing edge; follow the path in $G$ that
begins at this triangle.
\end{enumerate}
\end{definition}
This algorithm constitutes the standard proof that \sperner\ is indeed
a total search problem.
Figure~\ref{fig:redex} (RHS) shows the path leading to the unique solution
found by this algorithm. $p$ and its edges are not shown; in the
example the extra triangles do not occur in the path.
\end{paragraph}

\begin{theorem}\label{maintheorem}
It is \pspace-complete to compute the output of the natural path-following
algorithm, applied to instances of \twodb.
\end{theorem}

Theorem~\ref{maintheorem} is proved in the next section.
Before that, we note the following straightforward corollary.

\begin{corollary}\label{cor:sperner}
It is \pspace-complete to compute the output of the natural path-following
algorithm, applied to instances of \sperner.
\end{corollary}

\begin{proof}
Consider the reduction we described, from \twodb\ to \sperner.
Let $I$ be an instance of \twodb\ and $I'$ the corresponding instance of \sperner.
If $s$ is a triangle in $I'$ that lies in the square $(2x,2y),(2x+2,2y),(2x+2,2y+2),(2x,2y+2)$,
let $g(s)$ be the square $(x,y),(x+1,y),(x+1,y+1),(x,y+1)$. It can be checked that:
\begin{itemize}
\item if $s$ is a solution of $I'$ then $g(s)$ is a solution of $I$.
\item If $s$ and $s'$ are consecutive under the natural path-following algorithm
for \sperner, then either $g(s)=g(s')$, or $g(s)$ and $g(s')$ are consecutive
under the natural path-following algorithm for \twodb.
Conversely, if $g(s)$ and $g(s')$ are equal or consecutive, then there is
a short path from $s$ to $s'$.
\end{itemize}
This means that the path-following algorithm as given for \sperner\ mimics the
path-following algorithm as given for \twodb\, and if $s$ is the path-following solution of $I'$
then $g(s)$ is the path-following solution of $I$.
\qed
\end{proof}

\begin{paragraph}{Motivation, and some related work:}
It is known from \cite{CP95} that it is \pspace-complete to find the solution
of instances $(P,S)$ of \eotl\ consisting of the unique sink that is
connected to the given source {\bf 0} in $G(P,S)$.
That sink is the ``path-following solution'' of $(P,S)$.
(In~\cite{GPS11} this challenge is called {\sc Eotpl} for ``End of this particular line''.)
A well-known result of algorithmic game theory is the \ppad-completeness
of computing a Nash equilibrium of a given bimatrix game~\cite{cdt09,dgp09}.
Now, there are certain path-following algorithms that compute Nash equilibria, 
and these algorithms are of interest as a theory of equilibrium selection~\cite{HS}.
The best-known path-following algorithm for Nash equilibrium computation is the Lemke-Howson algorithm~\cite{LH64},
and it has been shown to be \pspace-complete to compute the outputs of Lemke-Howson
and related algorithms~\cite{GPS11}.

This paper is partly motivated by the ``paradox'' that the Lemke-Howson
algorithm is efficient in practice, although the computational complexity of its
solutions is much higher than unrestricted solutions, or even selected solutions
such as the lexicographically-first equilibrium.
Relevant to this is the possibility of a general principle,
saying that given any \ppad-complete problem $X$, and a path-following algorithm
${\cal A}$ for $X$, that the outputs of ${\cal A}$ are \pspace-complete to compute.

The problem under consideration, \twodb, represented a challenge instance for this conjecture.
To see why, it is helpful to contrast it with its three-dimensional analog, \threedb,
originally shown to be \ppad-complete by Papadimitriou~\cite{Pap94} in 1994.
Let $(f,g)$ be the reduction of~\cite{Pap94} from \eotl\ to \threedb, and let $(f',g')$
be the reduction from \threedb\ to \eotl\ that corresponds with the natural path-following
algorithm from \threedb\ (where we begin at the origin, and follow a sequence of cubes
that have colours 1,2, and 3, until colour 0 is encountered).
It can be checked that an \eotl\ graph $G(P,S)$ is topologically similar to the graph $G(f'(f(P,S)))$
in the sense that given any nodes $v_1$, $v_2$ in $G(P,S)$ at opposite ends of a
directed path, there exist nodes $v'_1$, $v'_2$ in $G(f'(f(P,S)))$ at opposite ends of
a directed path, such that $g(g'(v'_1))=v_1$ and $g(g'(v'_2))=v_2$.
Moreover, $g(g'({\bf 0}))={\bf 0}$, the directions of these paths remains the same, and the
correspondence between these pairs of endpoints is 1-1.
For our purposes, the main point to note is that if $v$ is the sink of $G(f'(f(P,S)))$
connected to ${\bf 0}$ then $g(g'(v))$ is the sink connected to {\bf 0} in $G(P,S)$.
It follows that it's \pspace-complete to compute the path-following solution to \threedb,
since the path-following solution of $G(P,S)$ is \pspace-complete to compute.

Now consider the 2009 \ppad-completeness proof of \twodb\ in~\cite{CD}.
Applying the above notation to their reduction, when we compare $G(P,S)$
with $G(f'(f(P,S)))$, the only thing that is preserved is the number of endpoints of paths.
In other respects, the topology is changed drastically, and in general if $v$ the vertex connected
to {\bf 0} in $G(f'(f(P,S)))$, it need not the case that $g(g'(v))$ is connected to {\bf 0} in $G(P,S)$.
This problem is incurable, since when we attempt to design a scheme
for embedding a large number of edges in the plane, there will typically be
crossing-points, and these crossing-points are handled by~\cite{CD} using a
gadget that changes the topology of the graph.

In our \pspace-completeness result for \twodb, it turns out to be convenient to reduce
from \qbf; this is in contrast with the \pspace-completeness results discussed above,
which reduce from the problem of computing the configuration reached by a space
bounded Turing machine after exponentially-many steps.
\end{paragraph}

\section{Proof of Theorem~\ref{maintheorem}}

We define a gadget that is used throughout the proof.

\begin{definition}\label{def:wire}
A {\em wire} consists of a sequence of grid points $p_1,p_2,...$ coloured 1,
together with a sequence of grid points $p'_1,p'_2,\ldots$ coloured 2,
with the following properties.
The distances between any $p_i$ and $p_{i+1}$, and between any $p'_i$ and
$p'_{i+1}$, is 1. If we move from $p_i$ to $p_{i+1}$, at least one point
on the right is some point $p'_j$, and least one point on the left is coloured 0;
similarly, if we move from $p'_i$ to $p'_{i+1}$, at least one point on the
left is some $p_j$ and at least one point on the right is coloured 0.
\end{definition}

Definition~\ref{def:wire} is designed to specify a (directed) path in the plane that must be
followed by the natural path-following algorithm for \twodb, assuming that
the path being followed separates regions coloured 1 and 2.
Figure~\ref{fig:wire} shows an example of a wire.

\begin{figure}
{\centering
\begin{tikzpicture}[scale=0.5]


\draw[black] (0,0) grid (8,8);
\foreach \y in {0,...,8}{
  \foreach \x in {0,...,8}{
    {\node at (\x,\y){\textcolor{white}{\LARGE $\bullet$}};{\node at (\x,\y){\textcolor{black}{\LARGE $\circ$}};
        }}}}

\foreach \x/\y in {1/0,1/1,1/2,1/3,1/4,1/5,1/6,2/6,3/6,4/6,5/6,5/5,5/4,5/3,6/3,7/3}{\node at (\x,\y){\textcolor{gray}{\LARGE $\bullet$}};}

\foreach \x/\y in {2/0,2/1,2/2,2/3,2/4,2/5,3/5,4/5,4/4,4/3,4/2,5/2,6/2,7/2}{\node at (\x,\y){\textcolor{black}{\LARGE $\bullet$}};}

\draw[black] (15,0) grid (23,8);

\draw[->,line width=3pt,rounded corners] (16.5,0)--(16.5,5.5)--(19.5,5.5)--(19.5,2.5)--(21.5,2.5);

\end{tikzpicture}
\par}
\caption{Wire example (left-hand side), and a simplified depiction (right-hand side) used in subsequent
diagrams, showing the directed path taken by the natural path-following algorithm.}\label{fig:wire}
\end{figure}
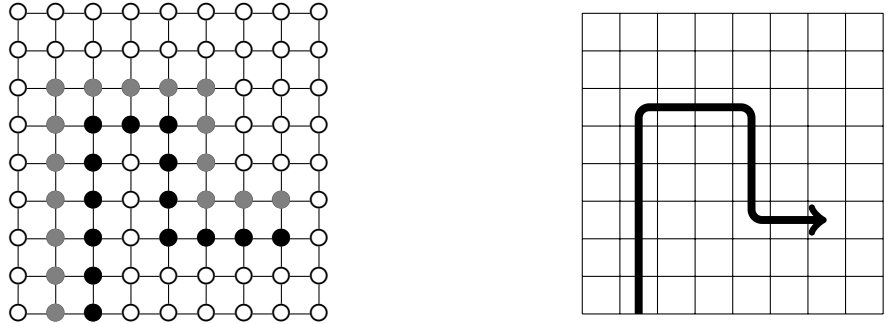

Suppose that an instance of \twodb\ is composed entirely of wires.
Notice that an endpoint of a wire ---other than the one at the origin---
is located at a solution of \twodb,
and any solution of \twodb\ is located at the endpoints of wires.

\bigskip
We reduce from \qbf, the problem of checking whether a given quantified boolean formula
with no free variables evaluates to true. Let
\[
\Phi = Q_1x_1Q_2x_2\ldots Q_nx_n\phi
\]
where $\phi$ is a propositional formula over variables $x_1\ldots x_n$
and each $Q_i$ is a quantifier. Define an $i$-subformula of $\Phi$ to
be a QBF derived from $\Phi$ by fixing the values of variables
$x_1,\ldots,x_i$. Thus there are $2^i$ $i$-subformulae of $\Phi$. A
0-subformula of $\Phi$ is just $\Phi$ itself.  If {\bf x} is a bit
string of length $i$ let $\Phi_{\bf x}$ be the $i$-subformula obtained
by setting $x_1,\ldots,x_i$ to {\bf x}.

Given $\Phi$, we construct an instance of \twodb\ given by a circuit $C_\Phi$ of size
polynomial in $n$, where $C_\Phi$ colours points in the integer grid having coordinates
that may be exponential in $n$.
Each $i$-subformula $\Phi_{\bf x}$ has an associated structure ${\cal S}(\Phi_{\bf x})$ in the planar grid.
Each such structure has an incoming wire on its left; on its right it has two outgoing wires,
with an incoming wire between them. The intention is that the
structure will have the following property:

\paragraph{Property 1:}{if $\Phi_{\bf x}$ evaluates to TRUE, then the
incoming wire on the left of ${\cal S}(\Phi_{\bf x})$ will link to the bottom outgoing wire on the
right, otherwise it will link to the top outgoing wire. The other two
wires will link together.}

\paragraph{Details of the construction of ${\cal S}(\Phi_{\bf x})$:}{
The location of ${\cal S}(\Phi_{\bf x})$ is as follows:
\begin{itemize}
\item ${\cal S}(\Phi_{\bf x})$ has a bounding box with height and width proportional to $2^{n-|{\bf x}|}$.
The coordinates of the bounding box are easy to compute from $\Phi$ and {\bf x}.
\item The bounding box of ${\cal S}(\Phi_{\bf x})$ contains the bounding boxes of
${\cal S}(\Phi_{{\bf x}0})$ and ${\cal S}(\Phi_{{\bf x}1})$, which are arranged side-by-side
(with a small gap between them) with their centres having the same $y$-coordinate.
(See Figure~\ref{fig:overall}).
\item For $|{\bf x}|=n$, structures ${\cal S}(\Phi_{\bf x})$ are arranged in a horizontal line in increasing order of {\bf x}.
\end{itemize}

Each bounding box has on its left-hand side, an ``incoming terminal'', a site
where a wire inside the bounding box ends at the left-hand side, and is directed
towards the interior of the box. On the right-hand side there are two outgoing terminals
(ends of wires directed outwards) and another incoming terminal between them
(See Figures~\ref{fig:bottomlevel}, \ref{fig:connectuniv}, \ref{fig:connectexist}).
The coordinates of these terminals are, of course, easy to compute from $\Phi$ and {\bf x}.

For $|{\bf x}|=n$, if {\bf x} satisfies $\phi$, ${\cal S}(\Phi_{\bf x})$
is as depicted on the LHS of Figure~\ref{fig:bottomlevel}, else as
shown on the RHS. Of course, satisfaction of $\phi$ can be checked
by a poly-sized boolean circuit that takes as input the coordinates of
the location of ${\cal S}(\Phi_{\bf x})$ and places the correct structure. Note that Property 1 is
satisfied by this design.

For $|{\bf x}|<n$, each $i$-subformula $\Phi_{\bf x}$ has structure ${\cal S}(\Phi_{\bf
 x})$ obtained by connecting up the structures ${\cal S}(\Phi_{{\bf
 x}0})$ and ${\cal S}(\Phi_{{\bf x}1})$ as described below.  By
construction these two sub-structures are consecutive and adjacent:
${\cal S}(\Phi_{{\bf x}0})$ just to the left of ${\cal S}(\Phi_{{\bf
 x}1})$.  Assume inductively that Property 1 is satisfied by these
structures.  We connect up their incoming and outgoing wires as shown
in Figure~\ref{fig:connectuniv} if $Q_{i+1}=\forall$ and as shown in
Figure~\ref{fig:connectexist} if $Q_{i+1}=\exists$.

The outermost structure ${\cal S}(\Phi)$ is connected by a wire starting from the known
source at the bottom left-hand corner of the domain, as shown in Figure~\ref{fig:overall}.
}

\medskip\noindent
It is useful in the subsequent proof to highlight the following observation:
\begin{observation}\label{obs:key}
Within any ${\cal S}(\Phi_{\bf x})$, the topological structure of the wires is that
each incoming terminal leads to an outgoing terminal, and there are no internal ends of wires.
\end{observation}

\paragraph{Proof that the connections encode the truth value of $\Phi$:}{
We claim that $\Phi$ evaluates to TRUE if and only if the wire that begins at
the origin ends at the lower outgoing terminal of ${\cal S}(\Phi)$, otherwise the wire
ends at the upper outgoing terminal (labelled NO in Figure~\ref{fig:overall}).

This is proved to hold for all structures ${\cal S}(\Phi_{\bf x})$ by backwards induction
on the length of {\bf x}. It can be seen to hold for $|{\bf x}|=n$, from Figure~\ref{fig:bottomlevel}.
We show that it holds for $|{\bf x}|=i<n$, assuming that it holds for all {\bf x} with $|{\bf x}|>i$.

Suppose first that $\Phi_{\bf x}$ is a universally quantified subformula, thus
$\Phi_{\bf x}=\forall x_{i+1} \cdots$. Thus $\Phi_{\bf x}$ is satisfied if $\Phi_{{\bf x}0}$ and
$\Phi_{{\bf x}1}$ are both satisfied.  We connect up ${\cal S}(\Phi_{{\bf x}0})$ and
${\cal S}(\Phi_{{\bf x}1})$ as shown in Figure~\ref{fig:connectuniv}.
Suppose $\Phi_{{\bf x}0}$ and $\Phi_{{\bf x}1}$ are indeed both satisfied.
By the inductive hypothesis, the incoming wire on the left
connects with the outgoing wire from ${\cal S}(\Phi_{{\bf x}0})$
labelled YES, which then connects via ${\cal S}(\Phi_{{\bf x}1})$
to the RHS outgoing YES wire.  If $\Phi_{{\bf x}0}$ is not
satisfied, the incoming wire connects to the outgoing NO wire of
${\cal S}(\Phi_{{\bf x}0})$, and thence directly to the RHS NO
wire.  If $\Phi_{{\bf x}0}$ is satisfied, but not $\Phi_{{\bf x}1}$,
then we reach the outgoing NO wire of ${\cal S}(\Phi_{{\bf x}1})$,
which feeds back to the RHS incoming wire of ${\cal S}(\Phi_{{\bf x}0})$,
and by Observation~\ref{obs:key}, the only way out of
${\cal S}(\Phi_{{\bf x}0})$ is via its outgoing NO wire (since
$\Phi_{{\bf x}0}$ is satisfied, the outgoing YES wire connects to
the incoming LHS wire), which takes us to the outgoing NO wire of
${\cal S}(\Phi_{\bf x})$.

The argument for the existential quantifier is similar, with respect to Figure~\ref{fig:connectexist}.
}

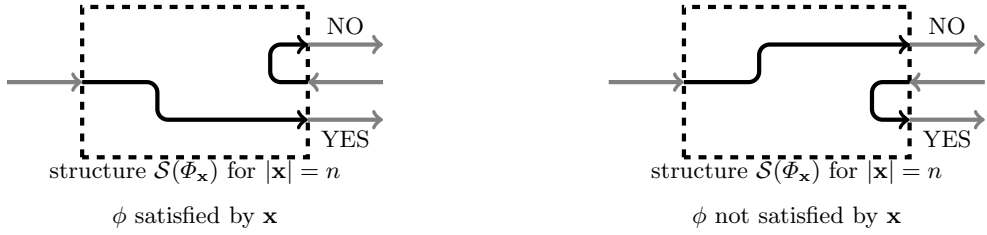
\begin{figure}
\begin{tikzpicture}

\draw[thick,white] (0,2)--(16,2)--(16,6)--(0,6)--(0,2);

\draw[->,gray,ultra thick] (1,4)--(2,4);
\draw[dashed,ultra thick] (2,3)--(2,5)--(5,5)--(5,3)--(2,3);
\draw[->,gray,ultra thick] (5,3.5)--(6,3.5);
\draw[<-,gray,ultra thick] (5,4.0)--(6,4.0);
\draw[->,gray,ultra thick] (5,4.5)--(6,4.5);
\node at (5.5,4.75) {NO};
\node at (5.5,3.25) {YES};
\draw[->,ultra thick,rounded corners] (2,4)--(3,4)--(3,3.5)--(5,3.5);
\draw[->,ultra thick,rounded corners] (5,4)--(4.5,4)--(4.5,4.5)--(5,4.5);
\node at (3.5,2.8) {structure ${\cal S}(\Phi_{\bf x})$ for $|{\bf x}|=n$};
\node at (3.5,2.2) {$\phi$ satisfied by {\bf x}};

\draw[->,gray,ultra thick] (9,4)--(10,4);
\draw[dashed,ultra thick] (10,3)--(10,5)--(13,5)--(13,3)--(10,3);
\draw[->,gray,ultra thick] (13,3.5)--(14,3.5);
\draw[<-,gray,ultra thick] (13,4.0)--(14,4.0);
\draw[->,gray,ultra thick] (13,4.5)--(14,4.5);
\node at (13.5,4.75) {NO};
\node at (13.5,3.25) {YES};
\draw[->,ultra thick,rounded corners] (10,4)--(11,4)--(11,4.5)--(13,4.5);
\draw[->,ultra thick,rounded corners] (13,4)--(12.5,4)--(12.5,3.5)--(13,3.5);
\node at (11.5,2.8) {structure ${\cal S}(\Phi_{\bf x})$ for $|{\bf x}|=n$};
\node at (11.5,2.2) {$\phi$ not satisfied by {\bf x}};

\end{tikzpicture}
\caption{Structures associated with length-$n$ bit vectors {\bf x}.
The dashed lines show the bounding boxes.
The black wires are internal to the structures, while the wires shown in grey are part
of larger structures ${\cal S}(\Phi_{\bf y})$ for {\bf y} a prefix of {\bf x}.}\label{fig:bottomlevel}
\end{figure}

\begin{figure}
\begin{tikzpicture}

\draw[white] (0,1) grid (16,7);

\draw[->,ultra thick] (2,4)--(3,4);
\draw[dashed,ultra thick] (3,3)--(3,5)--(6,5)--(6,3)--(3,3);
\draw[->,ultra thick] (6,3.5)--(7,3.5);
\draw[->,ultra thick] (6,4.5)--(7,4.5);
\node at (6.5,4.75) {NO};
\node at (6.5,3.25) {YES};

\draw[dashed,ultra thick] (9,3)--(9,5)--(12,5)--(12,3)--(9,3);
\draw[->,ultra thick] (12,3.5)--(14.5,3.5);
\draw[<-,ultra thick] (12,4.0)--(14.5,4.0);
\draw[->,ultra thick] (12,4.5)--(13,4.5);
\node at (12.5,4.75) {NO};
\node at (12.5,3.25) {YES};

\draw[dashed,ultra thick] (2,2)--(14.5,2)--(14.5,6.5)--(2,6.5)--(2,2);
\draw[->,ultra thick,rounded corners] (7,4.5)--(7.5,4.5)--(7.5,6)--(14,6)--(14,4.5)--(14.5,4.5);
\draw[->,ultra thick,rounded corners] (13,4.5)--(13.5,4.5)--(13.5,5.5)--(8,5.5)--(8,4)--(6,4);
\draw[->,gray,ultra thick] (15.5,4)--(14.5,4);
\draw[->,gray,ultra thick] (14.5,4.5)--(15.5,4.5);
\draw[->,gray,ultra thick] (14.5,3.5)--(15.5,3.5);
\draw[->,gray,ultra thick,rounded corners] (1,4)--(2,4);
\draw[->,ultra thick,rounded corners] (7,3.5)--(8.5,3.5)--(8.5,4)--(9,4);
\node at (15,4.75) {\bf NO};
\node at (15,3.25) {\bf YES};

\node at  (4.5,2.8) {structure ${\cal S}(\Phi_{{\bf x}0})$};
\node at (10.5,2.8) {structure ${\cal S}(\Phi_{{\bf x}1})$};
\node at  (7.5,1.8) {structure ${\cal S}(\Phi_{\bf x})$};
\end{tikzpicture}
\caption{combine 2 universally quantified sub-formulae:
${\cal S}(\Phi_{\bf x})$ connects the incoming and outgoing terminals of
${\cal S}(\Phi_{{\bf x}0})$ and ${\cal S}(\Phi_{{\bf x}1})$ using the
wires shown in black; wires shown in grey are part of a larger structure.
}\label{fig:connectuniv}
\end{figure}
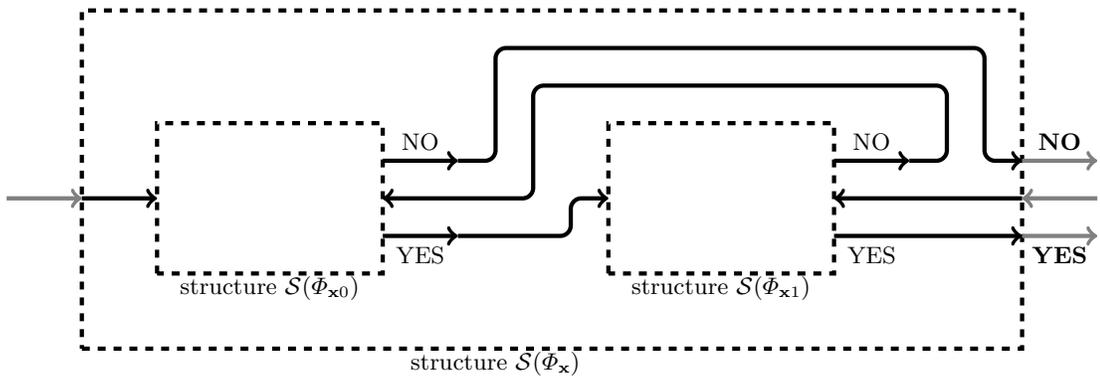

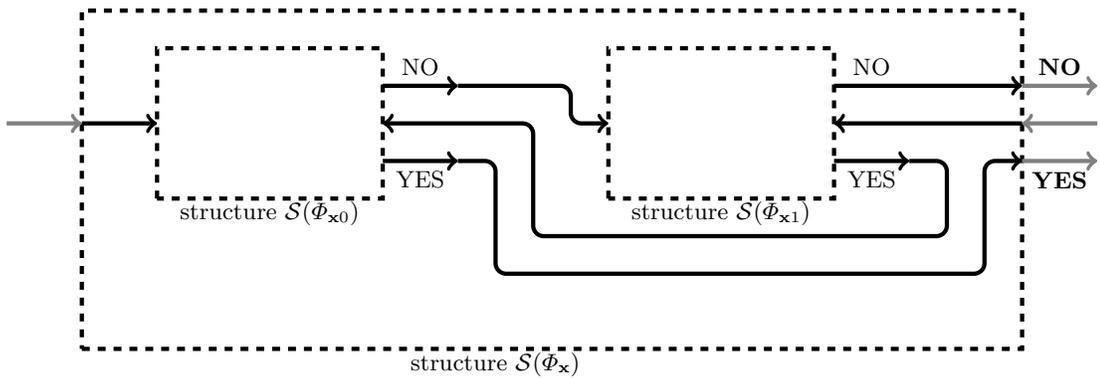
\begin{figure}
\begin{tikzpicture}

\draw[white] (0,0) grid (16,6);

\draw[->,ultra thick] (2,4)--(3,4);
\draw[dashed,ultra thick] (3,3)--(3,5)--(6,5)--(6,3)--(3,3);
\draw[->,ultra thick] (6,3.5)--(7,3.5);
\draw[->,ultra thick] (6,4.5)--(7,4.5);
\node at (6.5,4.75) {NO};
\node at (6.5,3.25) {YES};

\draw[dashed,ultra thick] (9,3)--(9,5)--(12,5)--(12,3)--(9,3);
\draw[->,ultra thick] (12,3.5)--(13,3.5);
\draw[->,ultra thick] (12,4.5)--(14.5,4.5);
\node at (12.5,4.75) {NO};
\node at (12.5,3.25) {YES};

\draw[dashed,ultra thick] (2,1)--(14.5,1)--(14.5,5.5)--(2,5.5)--(2,1);
\draw[->,ultra thick,rounded corners] (7,3.5)--(7.5,3.5)--(7.5,2)--(14,2)--(14,3.5)--(14.5,3.5);
\draw[->,ultra thick,rounded corners] (13,3.5)--(13.5,3.5)--(13.5,2.5)--(8,2.5)--(8,4)--(6,4);
\draw[->,gray,ultra thick] (15.5,4)--(14.5,4);
\draw[->,gray,ultra thick] (14.5,4.5)--(15.5,4.5);
\draw[->,gray,ultra thick,rounded corners] (1,4)--(2,4);
\draw[->,ultra thick,rounded corners] (7,4.5)--(8.5,4.5)--(8.5,4)--(9,4);
\draw[->,gray,ultra thick] (14.5,3.5)--(15.5,3.5);
\draw[->,ultra thick] (14.5,4)--(12,4);
\node at (15,4.75) {\bf NO};
\node at (15,3.25) {\bf YES};

\node at  (4.5,2.8) {structure ${\cal S}(\Phi_{{\bf x}0})$};
\node at (10.5,2.8) {structure ${\cal S}(\Phi_{{\bf x}1})$};
\node at  (7.5,0.8) {structure ${\cal S}(\Phi_{\bf x})$};
\end{tikzpicture}
\caption{combine 2 existentially quantified sub-formulae:
${\cal S}(\Phi_{\bf x})$ connects the incoming and outgoing terminals of
${\cal S}(\Phi_{{\bf x}0})$ and ${\cal S}(\Phi_{{\bf x}1})$ using the
wires shown in black; wires shown in grey are part of a larger structure.
}\label{fig:connectexist}
\end{figure}

\begin{figure}
\begin{tikzpicture}

\draw[dashed,ultra thick] (0,0)--(16,0)--(16,8)--(0,8)--(0,0);

\draw[->,ultra thick,rounded corners] (0,0)--(0.5,0.5)--(0.5,4)--(1,4);
\draw[dashed,ultra thick] (1,1)--(14.5,1)--(14.5,7)--(1,7)--(1,1);
\draw[->,ultra thick] (1,4)--(2,4);
\draw[dashed,ultra thick] (2,2)--(7,2)--(7,6)--(2,6)--(2,2);

\draw[dashed,ultra thick] (8,2)--(13,2)--(13,6)--(8,6)--(8,2);
\draw[->,ultra thick] (13,3.5)--(14.5,3.5);
\draw[<-,ultra thick] (13,4.0)--(14.5,4.0);

\draw[->,ultra thick,rounded corners] (7,4.5)--(7.2,4.5)--(7.2,6.7)--(14,6.7)--(14,4.5)--(14.5,4.5);
\draw[->,ultra thick,rounded corners] (13,4.5)--(13.5,4.5)--(13.5,6.3)--(7.4,6.3)--(7.4,4)--(7,4);
\draw[->,ultra thick] (15.5,4)--(14.5,4);
\draw[->,ultra thick] (14.5,4.5)--(15.5,4.5);
\draw[->,ultra thick] (14.5,3.5)--(15.5,3.5);
\draw[->,ultra thick,rounded corners] (7,3.5)--(7.6,3.5)--(7.6,4)--(8,4);
\node at (15,4.75) {\bf NO};
\node at (15,3.25) {\bf YES};

\node at  (4.5,1.8) {structure ${\cal S}(\Phi_{0})$};
\node at (10.5,1.8) {structure ${\cal S}(\Phi_{1})$};
\node at  (7.5,0.8) {structure ${\cal S}(\Phi)$};
\node at  (3.5,3.7) {${\cal S}(\Phi_{00})$};
\node at  (5.5,3.7) {${\cal S}(\Phi_{01})$};
\node at  (9.5,3.7) {${\cal S}(\Phi_{10})$};
\node at (11.5,3.7) {${\cal S}(\Phi_{11})$};

\foreach \x in {0,2,6,8}
{\draw[dashed,ultra thick] (3+\x,3)--(4+\x,3)--(4+\x,5)--(3+\x,5)--(3+\x,3);}

\foreach \x in {0,6}
{\draw[->,ultra thick](2+\x,4)--(3+\x,4);
\draw[->,ultra thick,rounded corners](4+\x,4.5)--(4.7+\x,4.5)--(4.7+\x,4)--(5+\x,4);
\draw[->,ultra thick,rounded corners](6+\x,3.5)--(6.5+\x,3.5)--(6.5+\x,2.7)
--(4.6+\x,2.7)--(4.6+\x,4)--(4+\x,4);
\draw[->,ultra thick,rounded corners](4+\x,3.5)--(4.3+\x,3.5)--(4.3+\x,2.3)--(6.7+\x,2.3)
--(6.7+\x,3.5)--(7+\x,3.5);
\draw[->,ultra thick,rounded corners](7+\x,4)--(6+\x,4);
\draw[<-,ultra thick,rounded corners](7+\x,4.5)--(6+\x,4.5);
}

\end{tikzpicture}
\caption{Encoding a formula $\Phi$ of the form $\forall x_1 \exists x_2 \cdots$ (not to scale).
The outermost dashed rectangle represents the entire domain of the corresponding
instance of \twodb. Recall that the bottom edge of the domain is coloured 2,
the left edge is coloured 1, and the other edges are coloured 0. Hence the
bottom-left square is known to contain colours 1 and 2. The arrows show the
wires connecting the (nested) structures that correspond to $\Phi$ itself,
$\Phi$ with $x_1$ set to 0 or 1, and $\Phi$ with the four alternative settings
of $x_1$ and $x_2$.\newline
The wire that begins at the origin will lead to the solution labelled YES if
$\Phi$ evaluates to TRUE, and NO if $\Phi$ evaluates to FALSE.
}\label{fig:overall}
\end{figure}
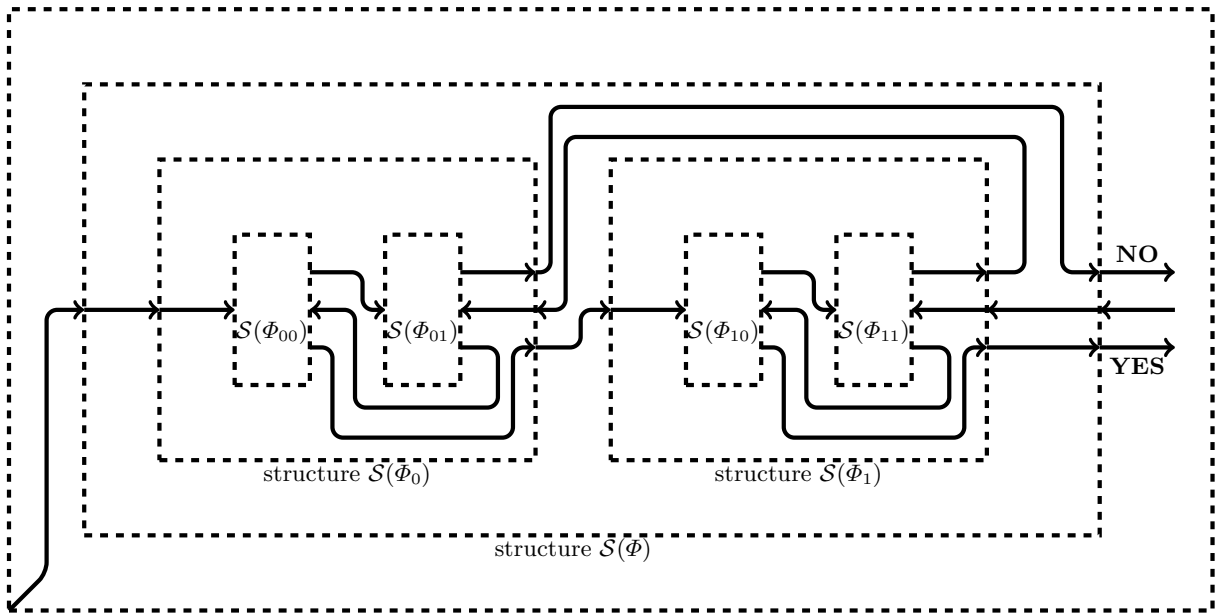

\FloatBarrier 

\section{Conclusion and Further Work}

With our current state of knowledge, it may be that for every \ppad-complete problem $X$,
and every path-following algorithm ${\cal A}$ for $X$, it is \pspace-complete to
compute the output of ${\cal A}$ on instances of $X$.
The reason to address this question is that it may shed light on the nature
of \ppad-completeness, perhaps helping to relate \ppad\ to other complexity classes.
We cannot go further and claim that it may be \pspace-complete to compute the
output of any exponential-time algorithm, since the output of lexicographic search has
complexity below \pspace.

A straightforward corollary of our main result is that if we are allowed to start at any
point on the boundary where two colours meet, and trace the path that begins there
until the third colour is reached, that the problem remains \pspace-complete.
This is reminiscent of the result of~\cite{GPS11}  that the problem of computing a
Lemke-Howson solution of a game remains \pspace-complete even if the algorithm is
free to choose the initially dropped label. We conjecture that these are instances of
a more general principle.


\bibliographystyle{splncs03}


\bibliography{literature}

\begin{thebibliography}{1}
\providecommand{\url}[1]{\texttt{#1}}
\providecommand{\urlprefix}{URL }

\bibitem{CD}
Chen, X., Deng, X.: {On the Complexity of 2D Discrete Fixed Point Problem}.
  Theoretical Computer Science  410(44),  4448--4456 (2009)

\bibitem{cdt09}
Chen, X., Deng, X., Teng, S.: Settling the complexity of computing two-player
  {N}ash equilibria. Journal of the ACM  56(3),  1--57 (2009)

\bibitem{CP95}
Crescenzi, P., Papadimitriou, C.: Reversible simulation of space-bounded
  computations. Theoretical Computer Science  143(1),  159--165 (1995)

\bibitem{dgp09}
Daskalakis, C., Goldberg, P.W., Papadimitriou, C.H.: {The complexity of
  computing a Nash equilibrium}. SIAM Journal on Computing  39(1),  195--259
  (2009)

\bibitem{GPS11}
Goldberg, P., Papadimitriou, C., Savani, R.: {The Complexity of the Homotopy
  Method, Equilibrium Selection, and Lemke-Howson Solutions}. In: Proceedings
  of 52nd FOCS. pp. 67--76 (2011)

\bibitem{HS}
Harsanyi, J., Selten, R.: A General Theory of Equilibrium Selection in Games.
  MIT Press (1988)

\bibitem{K88}
Krentel, M.: The complexity of optimization problems. Journal of Comput. System
  Sci  36(3),  490--509 (1988)

\bibitem{LH64}
Lemke, C., Howson~Jr., J.: Equilibrium points of bimatrix games. SIAM J. Appl.
  Math  12(2),  413--423 (1964)

\bibitem{Pap94}
Papadimitriou, C.: On the complexity of the parity argument and other
  inefficient proofs of existence. J. Comput. Syst. Sci.  48(3),  498--532
  (1994)

\end{thebibliography}

\end{document}